\documentclass[5p]{elsarticle}
\makeatletter
\def\ps@pprintTitle{%
 \let\@oddhead\@empty
 \let\@evenhead\@empty
 \def\@oddfoot{}%
 \let\@evenfoot\@oddfoot}
\makeatother

\usepackage{amsmath,amssymb,amsthm}
\usepackage{todonotes}

\usepackage[shortlabels]{enumitem}
\usepackage{url}

\newcommand{\real}{\ensuremath{\mathbb{R}}}

\newcommand{\smat}[1]{\ensuremath{\left[\begin{smallmatrix}#1\end{smallmatrix}\right]}}

\newcommand{\Qh}{\ensuremath{Q^{1/2}}}

\newtheorem{remark}{Remark}
\newtheorem{lemma}{Lemma}
\newtheorem{fact}{Fact}
\newtheorem{assumption}{Assumption}
\newtheorem{proposition}{Proposition}
\newtheorem{corollary}{Corollary}
\newtheorem{problem}{Problem}

\graphicspath{{./figures/}}

\allowdisplaybreaks

\begin{document}

\title{Data-based stabilization of unknown bilinear systems\\
with guaranteed basin of attraction\tnoteref{t1}}
\tnotetext[t1]{This research is partially supported by a Marie Sk\l{}odowska-Curie COFUND grant, no.~754315.}
\author[1]{Andrea Bisoffi\corref{cor1}%
}
\ead{a.bisoffi@rug.nl}
\author[1]{Claudio De Persis}
\ead{c.de.persis@rug.nl}
\author[2]{Pietro Tesi}
\ead{pietro.tesi@unifi.it}

\cortext[cor1]{Corresponding author}
\address[1]{ENTEG and the J.C. Willems Center for Systems and Control, University of Groningen, 9747 AG Groningen, The Netherlands}
\address[2]{DINFO, University of Florence, 50139 Florence, Italy}

\begin{abstract}
Motivated by the goal of having a building block in the direct design of data-driven controllers for nonlinear systems, we show how, 
for an unknown discrete-time bilinear system, the data collected in an offline open-loop experiment enable us to design a 
feedback controller and provide a guaranteed under-approximation of its basin of attraction. Both can be obtained by solving a linear matrix inequality for a fixed scalar parameter, and possibly iterating on different values of that parameter.
The results of this data-based approach are compared with the ideal case when the model is known perfectly.
\end{abstract}

\maketitle

\section{Introduction}
\label{sec:intro}

Direct data-driven control design aims at obtaining a control law based only on input-output data collected from the system during an experiment, thereby avoiding altogether the identification of a model of the system from the data.
Within the literature of direct data-driven control, notable approaches are iterative feedback tuning \cite{Hjalmarsson1998}, virtual-reference feedback tuning \cite{campi2002virtual}, iterative correlation-based tuning \cite{karimi2004iterative} to name a few.
When data are assumed to be generated by an underlying linear system, a number of approaches 
are available \cite{Salvador2018,Goncalves2019,baggioCSL19,recht2018tour,
deptCDC,Waarde2020}. However, sensibly fewer address the case when data are generated by an underlying nonlinear system, see \cite{fliess2009model,campi2006direct,D2IBC}.

To address the intrinsic difficulty of dealing with the control design of unknown nonlinear systems, a natural approach is to reduce their complexity by considering the system evolution along a given Lyapunov function. This classical control theoretic analysis is enhanced by nonparametric regression methods from machine learning to cope with the large uncertainty in the model \cite{Berkenkamp2017} and is performed using a sufficiently dense set of samples taken from the system. Analytical guarantees of stability and safety are then obtained relying on additional tools from robust control and optimization \cite{Wabersich2018ecc}. The approach of~\cite{Tabuada18data-driven,fliess2009model} to reduce the complexity of controlling unknown nonlinear systems consists of considering systems with a well-defined relative degree, in such a way that the uncertainty only appears in the form of two Lie derivatives of the output function along the system vector fields. Once the dynamics has been discretized, the key observation from sampled-data control theory is that these uncertain functions are constant between sampling times  for a sufficiently high sampling rate.

A different approach to data-driven control of nonlinear systems has been recently taken in a series of works that use the nonparametric representation of dynamical systems via Hankel matrices of finite-size input-output data proposed in~\cite{willems2005note}. On one hand, this representation has given rise to data-enabled predictive controllers where the effect of the nonlinearity is taken into account by a regularized optimization problem \cite{Coulson2018,Huang2019data}. On the other hand,  it inspired a data-dependent parametrization of the closed-loop system that reduces the control design to semidefinite programs  where the nonlinearity is dealt with as a process disturbance \cite{depersis-tesi2020tac}. Further results along this research thread have been proposed in \cite{berberich2019robust}. While these results make possible to deal with nonlinear systems, they provide local stability results. Very recently, within the research thread of~\cite[Thm.~1]{willems2005note}, there have been efforts to go beyond the local nature of the results for special classes of nonlinear systems, studying data-driven control of second-order discrete Volterra systems \cite{schiffer2020data} and polynomial systems \cite{guo20cdc}.

The goal of this paper is to characterize another notable class of nonlinear systems for which nonlocal data-driven control results can be established, namely 
bilinear systems.
The reason for focusing on bilinear systems is threefold.
In spite of their simple nonlinear structure, applying Carleman linearization to a generic continuous-time input-affine nonlinear system yields a 
continuous-time bilinear system with a larger state plus a remainder (see \cite{brockett1976volterra,krener1975bilinear}), so bilinear systems can be used as universal approximators of input-affine nonlinear systems \cite[p.~110]{rugh1981nonlinear}. This last consideration specifically motivates
the proposed data-driven control scheme for bilinear systems, which is envisioned to be a building block in future work on direct data-driven control of input-affine nonlinear systems (see also the discussion in Remark~\ref{rem:univ approx}).
A second motivation is to provide a method alternative to sum-of-squares programming for polynomial control systems \cite{guo20cdc} to directly design data-driven controllers of bilinear systems. 
Finally, bilinear systems are interesting {\it per se} as meaningful models for a number of relevant applications in engineering, biology and ecology \cite{mohler1973bilinear,bruni1974bilinear}.

Many model-based approaches have been proposed for control of  bilinear systems such as \cite{bitsoris2008ifac,tarbouriech2009control,Amato200c&sII,Khlebnikov2018arc}, and we refer the reader to \cite[\S 1]{Khlebnikov2018arc} for a thorough overview. 
Such model-based approaches assume the knowledge of the parameters of the bilinear system. When these are not known from first-principles considerations, one can resort to system identification techniques tailored for bilinear systems, and then apply one of the model-based approaches above. Some of these \emph{indirect} data-driven methods for system identification are \cite{favoreel1999subspace,chen2000subspace,
sontag2009input}, see also \cite[Part II]{verdult2002non} for an overview.
Although combining the aforementioned system identification techniques with model-based design constitutes a natural and valid way to control a bilinear system, we aim here at exploring the less-investigated direct control design of a bilinear system based on data (avoiding altogether a system identification step generally nontrivial in a nonlinear setting). 
We show that under mild assumptions (see Assumption~\ref{ass:matrixNormD} below), 
it is indeed possible to design stabilizing control policies directly from data.  
We also show via simulations that our 
approach compares well with a model-based design that has \emph{perfect} knowledge of the parameters of the system, regardless of whether this knowledge derives from first-principles considerations or from a preliminary system identification step.

In the case of data generated by an underlying linear system, the fundamental result \cite[Thm.~1]{willems2005note} has been shown in \cite{depersis-tesi2020tac} to allow direct data-driven design of (optimal) feedback controllers (with robustness to noise) for linear systems through linear matrix inequalities (LMI) \cite{boyd1994linear} and the \emph{local} stabilization of nonlinear systems through semidefinite programs. In the case of data generated by an underlying bilinear system, the arguments in \cite{depersis-tesi2020tac} need substantial modifications to counteract the nonlinear term appearing in the bilinear system and to explicitly provide an estimate of the region of attraction. Thus, we need to resort
to tools from robust control (such as \cite{petersen1987stabilization,Khlebnikov2018arc}, see Fact~\ref{fact:petersen lemma} below) besides more standard ones from linear matrix inequalities. 
Some conservatism is introduced in these steps compared to a model-based approach, as illustrated in Section~\ref{sec:numerical example}.

Similarly to the model-based approaches \cite{Khlebnikov2018arc,tarbouriech2009control} and, partially, to \cite{bitsoris2008ifac,Amato200c&sII}, we also adopt a linear state feedback and a quadratic Lyapunov function in the design of the closed-loop system. Alternatives are based on rational polynomial controllers and sum-of-squares programming \cite{Vatani2017ijrnc}. The  choice of linear controllers is restrictive compared to nonlinear state feedback (and the actual basin of attraction has not an ellipsoidal shape), but are dictated by the the desire of obtaining 
a computationally tractable result in the form of linear matrix inequalities (after fixing a scalar parameter). 
However, the main difference with those model-based approaches is that we design here the linear state feedback and the quadratic Lyapunov function \emph{without} relying on the knowledge of the bilinear system matrices, which we aim to substitute instead through data collected from the bilinear system.

Tuning a feedback controller based only on a limited number of open-loop data, which gives a guaranteed subset of the basin of attraction for a bilinear system, is the main contribution of this paper. 

\textbf{Structure.} The considered problem is formulated in Section~\ref{sec:description and problem}. In Section~\ref{sec:data-based} we provide our data-based controller for the unknown bilinear system with a guaranteed under-approximation of its basin of attraction, as a main result. Section~\ref{sec:numerical example} compares this solution with a model-based one on a numerical example.

\textbf{Notation.} For a matrix $A$, $\| A \|$ denotes the induced 2-norm.
For a symmetric matrix $\smat{A & B\\ B^\top & C}$, we may use the shorthand writing $\smat{A & B\\ \star & C}$. $I$ denotes an identity matrix of appropriate dimensions.

\section{System description and problem formulation}
\label{sec:description and problem}

Consider the discrete-time bilinear system
\begin{equation}
\label{bil.sys}
x^+ = A x + B u + D x u
\end{equation}
where $x \in \mathbb{R}^n$ is the state, $u \in \mathbb{R}$ is the input, and the system matrices have dimensions $A\in \mathbb{R}^{n\times n}$, $B\in \mathbb{R}^{n}$, $D\in \mathbb{R}^{n\times n}$. 
Our choice to consider a \emph{scalar} input in~\eqref{bil.sys} is motivated in Remark~\ref{rem:why scalar input} after we have outlined our approach.
The matrices $A$, $B$, $D$ are completely unknown apart from a bound on the matrix norm of $D$ as follows.

\begin{assumption}
\label{ass:matrixNormD}
The matrix $D$ satisfies $\| D \| \le \delta$ (equivalently, $D^\top D \preceq \delta^2 I$)
for some known $\delta>0$.
\end{assumption}

Our objective is to design a controller $u=K x$ for the bilinear system in~\eqref{bil.sys} based only on data collected from an off-line experiment (namely, without identifying the matrices $A$, $B$, $D$) and give a guaranteed under-approximation of the basin of attraction of the origin for the closed-loop system. The off-line experiment of duration $T$ (with $T>0$) collects the input and state sequences $u(0), u(1), \dots, u(T-1)$ and $x(0), x(1), \dots, x(T)$. These are organized as
\begin{subequations}
\label{data}
\begin{align}
U_{0,T} & := 
\begin{bmatrix}
u(0) & u(1) & \ldots & u(T-1)
\end{bmatrix} \label{U0}\\
X_{0,T} & := 
\begin{bmatrix}
x(0) & x(1) & \ldots & x(T-1)
\end{bmatrix} \label{X0}\\
X_{1,T} & := 
\begin{bmatrix} 
x(1) & x(2) & \ldots & x(T)
\end{bmatrix}, \label{X1}
\end{align}
and allow computing the auxiliary quantity
\begin{equation}
V_{0,T}:=
\begin{bmatrix}
x(0) u(0) & x(1) u(1) & \ldots & x(T-1) u(T-1)
\end{bmatrix}. \label{V0}
\end{equation}
\end{subequations}
Following \cite{depersis-tesi2020tac}, we reparametrize the gain $K$ by a matrix $G_K$ and give in the next lemma an equivalent representation of~\eqref{bil.sys} in closed loop with $u = K x$, which depends on data, except for the matrix $D$.
\begin{lemma}
\label{lem1}
Let $G_K \in \real^{T\times n}$ satisfy
\begin{equation}
\label{I=X0 GK}
I=X_{0,T} G_K.
\end{equation}
Then, system \eqref{bil.sys} with state feedback $u=Kx$ and $K= U_{0,T} G_K$ has the equivalent representation 
\begin{equation}
\label{bil.sys.data}
x^+ = (X_{1,T} - D V_{0,T} + Dx U_{0,T}) G_K x =: g_D(x) x.
\end{equation}
\end{lemma}
\begin{proof}
\eqref{bil.sys} with state feedback $u = K x$ becomes $x^+ = (A + B K + D x K) x$. This closed-loop matrix is, by~\eqref{I=X0 GK},
\begin{align*}
& A + B K + D x K = A \cdot I + B K + D x K \\
& = A X_{0,T} G_K + B U_{0,T} G_K + D x U_{0,T} G_K\\
& = (A X_{0,T} + B U_{0,T} + D x U_{0,T}) G_K\\
& = (X_{1,T} - D V_{0,T} + D x U_{0,T}) G_K,
\end{align*}
since the data in~\eqref{data} satisfy
$X_{1,T} = A X_{0,T} + B U_{0,T} + D V_{0,T}$, and this proves the statement.
\end{proof}
The reparametrization $G_K$ is a decision variable that we tune to achieve our control objective. Based on $G_K$ and on data, we define for compactness
\begin{equation}
\label{def cal quantities}
\mathcal{A}_c := X_{1,T} G_K,\, \mathcal{F} := I,\, \mathcal{H} := - V_{0,T} G_K,\, \mathcal{K} := U_{0,T} G_K,
\end{equation}
so that the closed-loop representation in~\eqref{bil.sys.data} becomes
\begin{equation}
\label{bilinear data D}
x^+ = (\mathcal{A}_c + \mathcal{F} D \mathcal{H} + D x \mathcal{K} )x = g_D(x) x,
\end{equation}
where $D$ is highlighted. As mentioned earlier, we aim at giving a guaranteed under-approximation of the basin of attraction of the closed-loop system in~\eqref{bilinear data D}. We do so by considering a quadratic Lyapunov function
\begin{equation}
\label{Lf}
V(x)=x^\top Q x
\end{equation}
with $Q=Q^\top \succ 0$ and imposing the strict decrease of $V\big(g_D(x)x\big) - V(x)$ for the dynamics in~\eqref{bilinear data D}. The last quantity is easily computed as in the next claim.
\begin{lemma}
\label{lemma:Lf increment}
We have that
$V\big( g_D(x)x \big) - V(x) = x^\top \mathcal{N}_D(x) x$ with $\mathcal{N}_D(x)$ defined as
\begin{align}
& \mathcal{N}_D(x):= (\mathcal{A}_c + \mathcal{F} D \mathcal{H})^\top Q (\mathcal{A}_c + \mathcal{F} D \mathcal{H}) - Q  \notag \\
& +(\mathcal{A}_c + \mathcal{F} D \mathcal{H})^\top Q Dx \mathcal{K} +  \mathcal{K}^\top x^\top D^\top Q (\mathcal{A}_c + \mathcal{F} D \mathcal{H})\notag \\
& + \mathcal{K}^\top x^\top D^\top Q D x \mathcal{K}. \label{N_D(x)}
\end{align}
\end{lemma}
\begin{proof}
The expression for $\mathcal{N}_D(x)$ is immediate by substituting \eqref{bilinear data D} in $V\big( g_D(x)x \big) - V(x)$.
\end{proof}
Note that for $D=0$, \eqref{bil.sys} becomes linear and \eqref{N_D(x)} reduces to $\mathcal{N}_D(x) = \mathcal{A}_c^\top Q \mathcal{A}_c - Q$,  corresponding to the classical Lyapunov condition for discrete-time linear systems. 
We impose $V\big(g_D(x)x\big) - V(x) < 0$ for all $x \neq 0$ in the ellipsoid 
\begin{equation}
\label{E_Q}
\mathcal{E}_Q := \{ x \in \real^n \colon x^\top Q x \le 1\},
\end{equation}
by designing the decision variables $G_K$, which determines $K= U_{0,T} G_K$, and $Q$, which will be optimized so that the volume of the ellipsoid $\mathcal{E}_Q$ is maximized. The design will be based only on data, and return the ellipsoid $\mathcal{E}_Q$ as a guaranteed under-approximation of the basin of attraction.

We summarize the system description and control objective illustrated in this section as follows.
\begin{problem}
\label{probl:state}
Based only on data collected from an off-line experiment as in the matrices in~\eqref{data}, obtain a controller $u= K x$ for \eqref{bil.sys}
such that for the closed-loop system, the origin has a guaranteed basin of attraction. The data-based design is performed based on the decision variables $G_K$ and $Q$ of the quadratic Lyapunov function in~\eqref{Lf}, a sublevel set of which gives the guaranteed basin of attraction.
\end{problem}

Some remarks are in order.

\begin{remark}
The existence of a matrix $G_K$ satisfying \eqref{I=X0 GK}
is related to the ``quality" of the experimental data. In fact,
condition \eqref{I=X0 GK} expresses the property that 
the data are sufficiently rich so that the system dynamics 
can be parametrized directly in terms of the matrices in \eqref{data}. 
A key property established in \cite{willems2005note}
is that, for linear systems, $X_{0,T}$ is full-row rank (thus, a solution $G_K$ to~\eqref{I=X0 GK} exists) when the experiment is carried out using a sufficiently {exciting}
input signal. An extension of this property to nonlinear systems
is discussed in \cite{deptMTNS2020} where it is shown that 
under prior knowledge of an upper bound on the nonlinearity 
(in fact, on $D$ in the present case of bilinear systems)
one can always design experiments so that \eqref{I=X0 GK} is fulfilled. 
\end{remark}

\begin{remark}
\label{rem:why scalar input}
The present analysis can be extended to bilinear systems 
with input $u \in \real^m$ and $m\ge 2$.
For $m=2$, \eqref{bil.sys} can be written for $u=(u_1,u_2)$ as
\begin{equation}
\label{bil.sys.2d.input}
x^+ = A x + B_1 u_1 + B_2 u_2 + D_1 x u_1 + D_2 x u_2.
\end{equation}
We can define $U_{0,T}^{(1)}$ and $U_{0,T}^{(2)}$ as in \eqref{U0}, but considering respectively the components $u_1$ and $u_2$. Similarly, we can define $V_{0,T}^{(1)}$ and $V_{0,T}^{(2)}$ as in~\eqref{V0}. 
Based on the very same steps as in Lemma~\ref{lem1}, we can obtain for $U_{0,T}=\smat{U_{0,T}^{(1)}\\U_{0,T}^{(2)}}$ the next equivalent representation of~\eqref{bil.sys.2d.input}
\[
x^+ \!= (X_{1,T}\! - D_1 V_{0,T}^{(1)}\! - D_2 V_{0,T}^{(2)}\! + D_1 x U_{0,T}^{(1)} \! + D_2 x U_{0,T}^{(2)} ) G_K x.
\]
This expression shows by comparison with~\eqref{bil.sys.data} that the case for $m=2$ can be treated using the same procedure we develop in the presence of a single unknown $D$, and this consideration easily generalizes to $m$ larger than $2$. 
For this reason we focus on the essential case with input $u \in \real$.%
\end{remark}%

\begin{remark}
\label{rem:univ approx}
The universal approximation property of bilinear systems mentioned in Section~\ref{sec:intro} holds with respect to continuous-time nonlinear systems. We focus here on discrete-time bilinear systems since the data in~\eqref{data} are samples obtained from experiments. However, the same results would apply to continuous-time bilinear systems if $X_{1,T}$ in~\eqref{X1} is replaced by the samples of the state derivatives. These results would then lend themselves to the analysis of a bilinear approximation of continuous-time nonlinear systems (provided disturbances are accounted for).
\end{remark}

\section{Data-based solution with guaranteed basin of attraction}
\label{sec:data-based}

In Section~\ref{sec:description and problem}, we showed that data allow expressing \eqref{bil.sys} in closed loop with $u=K x$ as \eqref{bilinear data D} (by introducing the reparametrization $G_K$ of $K$). Data, however, did not allow us to completely remove the matrices of model~\eqref{bil.sys}. In particular, $g_D(x)$ in~\eqref{bilinear data D} still contains two instances of matrix $D$ (namely, $D x \mathcal{K}$ and $\mathcal{F} D \mathcal{H}$), which can both be interpreted as perturbation of the matrix $\mathcal{A}_c$. 
In this section we first address the former, which is more standard and occurs analogously for model-based design of a bilinear system (see, e.g., \cite{Khlebnikov2018arc}), and then the latter, which is motivated by our desire to solve Problem~\ref{probl:state} based only on data and calls for the matrix norm bound in Assumption~\ref{ass:matrixNormD}. 

Before presenting the developments of this section, we recall 
an auxiliary result from~\cite{petersen1987stabilization}, which has been reported in a convenient form as \cite[Lemma~1]{Khlebnikov2018arc} and is related to the S-procedure \cite[\S 2.6.3]{boyd1994linear}. In particular, \cite[Lemma~1]{Khlebnikov2018arc} implies the next fact.
\begin{fact} \label{fact:petersen lemma}
(\cite[Lemma~1]{Khlebnikov2018arc}) Let $\mathsf{G}= \mathsf{G}^\top \in \real^{\mathsf{n} \times \mathsf{n}}$, $\mathsf{M} \in \real^{\mathsf{n} \times \mathsf{p}}$, $\mathsf{N} \in \real^{\mathsf{n} \times \mathsf{q}}$.
\begin{equation}
\label{petersen hypothesis}
\begin{aligned}
& \mathsf{G}+\mathsf{M} \mathsf{D} \mathsf{N}^\top + \mathsf{N} \mathsf{D}^\top \mathsf{M}^\top \prec 0\\
& \hspace*{2.5cm} \text{for all } \mathsf{D} \in \real^{\mathsf{p} \times \mathsf{q}} \text{ with } \| \mathsf{D} \| \le 1
\end{aligned}
\end{equation}
if there exists a scalar $\mathsf{e}$ such that
\begin{equation}
\label{petersen thesis}
\begin{bmatrix}
\mathsf{G} + \mathsf{e} \mathsf{M} \mathsf{M}^\top 	& \mathsf{N}\\
\mathsf{N}^\top 									& - \mathsf{e} I\\
\end{bmatrix} \prec 0.
\end{equation}
\end{fact}

\begin{remark}
\cite{Amato200c&sII} considers a model-based setting similar to \cite{Khlebnikov2018arc} returning a polytope as a subset of the basin of attraction. However, since \cite[Thm.~2]{Amato200c&sII} does so by dilating the polytope and including in that a quadratic sublevel set, we refer to the naturally quadratic approach in~\cite{Khlebnikov2018arc}.
\end{remark}
With Fact~\ref{fact:petersen lemma} we are in a position to develop this section. The next lemma addresses the term $D x \mathcal{K}$ in $g_D(x)$ in~\eqref{bilinear data D}. Specifically, it shows that as long as we restrict the analysis to a sublevel set $\mathcal{E}_Q$ (defined in \eqref{E_Q}) of the Lyapunov function $V$ in~\eqref{Lf} (where $Q$ itself is a decision variable determining the size of this sublevel set), strict decrease of $V$ along solutions is guaranteed ($\mathcal{N}_D(x) \prec 0$) since $\mathcal{N}_D(x)$ determines $V\big( g_D(x)x \big) - V(x)$ as in Lemma~\ref{lemma:Lf increment}.
\begin{lemma}
\label{lemma:mi with D}
If there exist
$\tau$ and $Q=Q^\top$ such that
\begin{align}
& \begin{bmatrix}
- Q 												& 0			&  \mathcal{K}^\top	& (\mathcal{A}_c+\mathcal{F} D\mathcal{H})^\top\\
0													& - \tau Q 	& 0						& D^\top\\
\mathcal{K}										& 0			& -\tfrac{1}{\tau} I	& 0\\
\mathcal{A}_c+\mathcal{F} D \mathcal{H} 			& D 		& 0						& - Q^{-1} 
\end{bmatrix} \prec 0, \label{mi with D}
\end{align}
then $\mathcal{N}_D(x) \prec 0$ for all $x \in \mathcal{E}_Q$.
\end{lemma}

\begin{proof}
The proof follows closely \cite{Khlebnikov2018arc}, but is reported for self-containedness. Define for compactness
\begin{equation}
\label{cal Q}
\mathcal{R}:=-Q + (\mathcal{A}_c + \mathcal{F} D \mathcal{H})^\top Q (\mathcal{A}_c + \mathcal{F} D \mathcal{H})
\end{equation}
and note for the following that \eqref{mi with D} implies $Q\succ 0$ and $\tau>0$.
By Schur's complement (with respect to lowest block $-Q^{-1}$), \eqref{mi with D} is equivalent, by \eqref{cal Q}, to
\begin{equation*}
\begin{bmatrix}
\mathcal{R} 		& (\mathcal{A}_c + \mathcal{F} D \mathcal{H})^\top Q D	& \mathcal{K}^\top \\
\star				& -\tau Q + D^\top Q D								& 0\\
\star				& \star												& -\tfrac{1}{\tau} I
\end{bmatrix} \prec 0.
\end{equation*}
By Schur's complement, this inequality is equivalent to
\begin{equation*}
\begin{bmatrix}
\mathcal{R}	& (\mathcal{A}_c + \mathcal{F} D \mathcal{H})^\top Q D	& \mathcal{K}^\top  	& 0\\
\star		& -\tau Q 												& 0						& D^\top Q\\
\star		& \star													& -\tfrac{1}{\tau} I 	& 0\\
\star		& \star													& \star					& -Q
\end{bmatrix} \prec 0.
\end{equation*}
Rearranging rows and columns of this inequality gives
\begin{equation*}
\begin{bmatrix}
\mathcal{R}	& 0			& (\mathcal{A}_c + \mathcal{F} D \mathcal{H})^\top Q D	& \mathcal{K}^\top  \\
\star 					& -Q		& Q D													& 0					\\
\star					& \star 	& -\tau Q 												& 0					\\
\star					& \star		& \star													& -\tfrac{1}{\tau} I \\
\end{bmatrix} \prec 0,
\end{equation*}
which is equivalent to \eqref{mi with D}.
We want to put this inequality in a form where we can apply Fact~\ref{fact:petersen lemma}. Then, we pre- and postmultiply the previous inequality by the block diagonal matrix with entries $I$, $I$, $(\Qh)^{-1}$, $I$ (where $\Qh$ is the unique symmetric, positive definite square root matrix for ${Q}={Q}^\top \succ 0$ \cite[Thm.~7.2.6]{horn2013matrix}, so that ${Q}=\Qh \Qh$) and apply Schur's complement (with respect to the lowest block $-\tfrac{1}{\tau} I$) to obtain with some computations
\begin{equation*}
\begin{small}
\hspace*{-3pt}
\left[
\begin{array}{@{}c|c@{}}
\begin{bmatrix}
\mathcal{R}	& \hspace*{-5pt}0\\
0			& \hspace*{-5pt}-Q
\end{bmatrix}
\hspace*{-2pt}
+ \hspace*{-2pt}
\tau
\begin{bmatrix}
\mathcal{K}^\top \\
0
\end{bmatrix}\hspace*{-3pt}
\begin{bmatrix}
\mathcal{K} & \hspace*{-4pt}0
\end{bmatrix}
& \begin{bmatrix}
(\mathcal{A}_c \!+\! \mathcal{F} D \mathcal{H})^\top Q D (\Qh)^{-1}\\
Q D (\Qh)^{-1}
\end{bmatrix}\\ \hline
\star & - \tau I\
\end{array}
\right] \prec 0.
\end{small}
\end{equation*}
Note that $x^\top Q x = (x^\top \Qh)(\Qh x)$, hence for all $x$ such that $x^\top Q x \le 1$, $\| x^\top \Qh \| \le 1$. With this observation and by Fact~\ref{fact:petersen lemma} we conclude (after some simplifications) that
\begin{align}
& \begin{bmatrix}
\mathcal{R}	& 0\\
0			& -Q
\end{bmatrix} 
+
\begin{bmatrix}
\mathcal{K}^\top\\
0
\end{bmatrix}
x^\top
\begin{bmatrix}
D^\top Q (\mathcal{A}_c + \mathcal{F} D \mathcal{H})												& D^\top Q
\end{bmatrix}
\notag \\
&
\hspace*{2cm}
+
\begin{bmatrix}
(\mathcal{A}_c + \mathcal{F} D \mathcal{H})^\top Q D\\
Q D
\end{bmatrix}
x
\begin{bmatrix}
\mathcal{K} & 0
\end{bmatrix}
\prec 0 \label{Delta V.before final Schur}
\end{align}
for all $x$ such that $x^\top Q x \le 1$. We show now that this is equivalent to the conclusion of the lemma.
Define for compactness
\[
\mathcal{P} := \mathcal{R} +(\mathcal{A}_c + \mathcal{F} D \mathcal{H})^\top Q Dx \mathcal{K} + \mathcal{K}^\top x^\top D^\top Q (\mathcal{A}_c + \mathcal{F} D \mathcal{H}),
\]
so that \eqref{Delta V.before final Schur}  is equivalent, after some computations, to
\begin{equation*}
\begin{bmatrix}
\mathcal{P} & \mathcal{K}^\top x^\top D^\top Q\\
Q D x \mathcal{K} & -Q\\
\end{bmatrix} \prec 0.
\end{equation*}
By Schur's complement, we obtain that
\begin{equation*}
\mathcal{P} +  \mathcal{K}^\top x^\top D^\top Q D x \mathcal{K} \prec 0 \quad \text{ for all } x \text{ such that } x^\top Q x \le 1,
\end{equation*}
which is equivalent by~\eqref{N_D(x)} to $\mathcal{N}_D(x) \prec 0$ for all $x \in \mathcal{E}_Q$.
\end{proof}

The next lemma addresses the term $\mathcal{F} D \mathcal{H}$ in $g_D(x)$ in~\eqref{bilinear data D}. Specifically, it shows that as long as the matrix $D$ is bounded in norm by $\delta$ as in Assumption~\ref{ass:matrixNormD}, we can obtain a matrix inequality depending only on $\delta$ and  guarantee that Lemma~\ref{lemma:mi with D} and its conclusions hold for all such $D$, which is key to obtain a fully data-based solution to our problem.

\begin{lemma}
\label{lemma:mi without D}
Let Assumption~\ref{ass:matrixNormD} hold.
If there exist $\tau$, $\epsilon_2$ and $Q=Q^\top$ such that
\begin{equation}
\begin{bmatrix}
- Q 					& 0				& \mathcal{K}^\top		 	& \mathcal{A}_c^\top 		& \delta \mathcal{H}^\top \\
0						& - \tau Q 		& 0							& 0							& \delta I\\
\mathcal{K}				& 0				& -\tfrac{1}{\tau} I		& 0							& 0\\
\mathcal{A}_c			& 0				& 0							& - Q^{-1} + \epsilon_2 I	& 0\\
\delta \mathcal{H} 		& \delta I		& 0							& 0 						& - \epsilon_2 I
\end{bmatrix} \prec 0,
\label{mi without D}
\end{equation}
then \eqref{mi with D} holds.
\end{lemma}
\begin{proof}
Note that from $\mathcal{F}=I$ in~\eqref{def cal quantities}, \eqref{mi with D} is equivalent to
\begin{align*}
& \begin{bmatrix}
- Q 					& 0				& \mathcal{K}^\top		& \mathcal{A}_c^\top 		\\
0						& - \tau Q 		& 0						& 0							\\
\mathcal{K} 			& 0				& -\tfrac{1}{\tau} I	& 0							\\
\mathcal{A}_c 			& 0				& 0						& - Q^{-1} 					\\
\end{bmatrix} + 
\begin{bmatrix}
0\\ 0\\ 0\\ \mathcal{F}
\end{bmatrix}
 \frac{D}{\delta} 
\begin{bmatrix}
\delta \mathcal{H}  & \delta I & 0 & 0 
\end{bmatrix}\notag \\
& \hspace*{.2cm}
+ \begin{bmatrix}
\delta \mathcal{H}^\top \\ \delta I \\ 0 \\ 0 
\end{bmatrix}
 \frac{D^\top}{\delta} 
\begin{bmatrix}
0& 0& 0& \mathcal{F}^\top
\end{bmatrix} \prec 0
\end{align*}
and this equation has the same structure as $\mathsf{G}+\mathsf{M} \mathsf{D} \mathsf{N}^\top + \mathsf{N} \mathsf{D}^\top \mathsf{M}^\top \prec 0$ in Fact~\ref{fact:petersen lemma}, since $\| D \| \le \delta$ ($\delta >0$) by Assumption~\ref{ass:matrixNormD}. 
Indeed, by making the suitable correspondences between the quantities of this lemma and those of Fact~\ref{fact:petersen lemma}, the existence of $\epsilon_2$ such that \eqref{mi without D} holds (corresponding to~\eqref{petersen thesis} of Fact~\ref{fact:petersen lemma}) guarantees that \eqref{mi with D} (corresponding to~\eqref{petersen hypothesis} of Fact~\ref{fact:petersen lemma}) holds for $D$ as in Assumption~\ref{ass:matrixNormD}.
\end{proof}
Lemma~\ref{lemma:mi without D} enables us to generalize the conclusions of Lemma~\ref{lemma:mi with D} for all $D$ with $\| D \| \le \delta$, so that we do not need to rely on the knowledge of $D$ (as it would be the case in a model-based scheme), but just on its (possibly loose) norm bound $\delta$. The matrix inequality \eqref{mi without D} of Lemma~\ref{lemma:mi without D} (where only $\delta$ appears), however, contains products of decision variables and inverses of decision variables. 
We address this in the next proposition, which obtains a matrix inequality that is as close as possible to an LMI (hence efficient to solve) and expresses explicitly the matrix inequality in terms of the available data. This proposition is the main result of this paper.
\begin{proposition}
\label{prop:data based quadr stab}
Under Assumption~\ref{ass:matrixNormD}, suppose there exist $\epsilon_1\in \real$, $\epsilon_2\in \real$, $Y \in \real^{n \times T}$ and $P=P^\top\in \real^{n \times n}$ such that%
\begin{subequations}
\label{qlmi data:all}
\begin{align}
& \begin{bmatrix}
- P 					& 0						& Y U_{0,T}^\top 			& Y X_{1,T}^\top 			& -\delta Y V_{0,T}^\top \\
0						& - \epsilon_1 P 		& 0							& 0							& \delta \epsilon_1 P\\
U_{0,T} Y^\top 			& 0						& -\epsilon_1 I				& 0							& 0\\
X_{1,T} Y^\top			& 0						& 0							& - P + \epsilon_2 I		& 0\\
-\delta V_{0,T} Y^\top	& \delta \epsilon_1 P	& 0							& 0 						& - \epsilon_2 I
\end{bmatrix} \prec 0 \label{qlmi data}\\
& P=X_{0,T} Y^\top, \label{qlmi data:I=X0 GK}
\end{align}
\end{subequations}
and set $Q= P^{-1}$, $G_K= Y^\top P^{-1}$. Then,
\begin{enumerate}[wide,labelindent=0pt,label=(\roman*)]
\item \label{prop:data based quadr stab:lyap} for the dynamics in~\eqref{bilinear data D} corresponding to $D$, the Lyapunov function $V(x)= x^\top Q x= x^\top (X_{0,T} Y^\top)^{-1} x$ satisfies 
\begin{equation*}
V(g_D(x) x) - V(x) < 0 \quad \text{ for all } x \in \mathcal{E}_{Q}\backslash\{0\};
\end{equation*}
\item \label{prop:data based quadr stab:AS} the origin is asymptotically stable for \eqref{bil.sys} with controller $u=K x = U_{0,T} G_K x = U_{0,T} Y^\top (X_{0,T} Y^\top)^{-1} x$ and its basin of attraction contains the set $\mathcal{E}_{Q}$.
\end{enumerate}
\end{proposition}
\begin{proof}
We begin showing that inequalities \eqref{mi without D} and \eqref{qlmi data} are equivalent, noting for the following that \eqref{qlmi data} implies $P \succ 0$.
With the definitions in~\eqref{def cal quantities}, \eqref{mi without D} is equivalent to
\begin{equation*}
\begin{bmatrix}
- Q 					& 0				& G_K^\top U_{0,T}^\top  	& G_K^\top X_{1,T} ^\top	&  - \delta G_K^\top V_{0,T}^\top \\
\star					& - \tau Q 		& 0							& 0							& \delta I\\
\star					& \star			& -\tfrac{1}{\tau} I		& 0							& 0\\
\star					& \star			& \star						& - Q^{-1} + \epsilon_2 I	& 0\\
\star					& \star			& \star						& \star						& - \epsilon_2 I
\end{bmatrix} \prec 0.
\end{equation*}
By pre- and postmultipling this inequality by the block diagonal matrix with entries $Q^{-1}$, $Q^{-1}$, $I$, $I$, $I$ and by setting $Q= P^{-1}$, $G_K= Y^\top P^{-1}$ as in the statement of the proposition, the last inequality is equivalent to
\begin{equation*}
\begin{bmatrix}
- P 					& 0				& Y U_{0,T}^\top  		& Y X_{1,T} ^\top		&  - \delta Y V_{0,T}^\top \\
\star					& - \tau P 		& 0						& 0						& \delta P\\
\star					& \star			& -\tfrac{1}{\tau} I	& 0						& 0\\
\star					& \star			& \star					& - P + \epsilon_2 I	& 0\\
\star					& \star			& \star					& \star					& - \epsilon_2 I
\end{bmatrix} \prec 0.
\end{equation*}
To avoid the simultaneous presence of $\tau$ and $1/\tau$, this inequality is equivalent to the next one by pre- and postmultiplying by  the block diagonal matrix with entries $I$, $\tfrac{1}{\tau}I$, $I$, $I$, $I$ and setting $\epsilon_1=1/\tau$:
\begin{equation*}
\begin{bmatrix}
- P 					& 0						& Y U_{0,T}^\top  		& Y X_{1,T} ^\top		&  - \delta Y V_{0,T}^\top \\
\star					& - \epsilon_1 P		& 0						& 0						& \delta \epsilon_1 P\\
\star					& \star					& -\epsilon_1 I			& 0						& 0\\
\star					& \star					& \star					& - P + \epsilon_2 I	& 0\\
\star					& \star					& \star					& \star					& - \epsilon_2 I
\end{bmatrix} \prec 0,
\end{equation*}
which is exactly \eqref{qlmi data}. After these manipulations, the conclusions of the proposition follow readily. Indeed, the fact that \eqref{qlmi data} holds, implies that \eqref{mi without D} holds, and then by Lemmas~\ref{lemma:mi without D} and \ref{lemma:mi with D} that 
$D$ as in Assumption~\ref{ass:matrixNormD} satisfies $\mathcal{N}_D(x) \prec 0$ for all $x \in \mathcal{E}_Q$.
By Lemma~\ref{lemma:Lf increment}, \ref{prop:data based quadr stab:lyap} follows.
\eqref{qlmi data:I=X0 GK}, which is equivalent to $I=X_{0,T} G_K$, and Lemma~\ref{lem1} ensure that
\eqref{bil.sys.data} or, equivalently, \eqref{bilinear data D} are an equivalent representation of \eqref{bil.sys} with controller $u=K x = U_{0,T} G_K x$. Standard Lyapunov theorems give then \ref{prop:data based quadr stab:AS}.
\end{proof}

Proposition~\ref{prop:data based quadr stab} effectively solves Problem~\ref{probl:state}. Indeed, if a solution to \eqref{qlmi data:all} is found (which is based on data from an off-line experiment), then we have a controller $K$ and a guaranteed basin of attraction in terms of the set $\mathcal{E}_Q$.

The matrix inequality \eqref{qlmi data} in Proposition~\ref{prop:data based quadr stab} is convenient because, after fixing the scalar $\epsilon_1$, it is an LMI in the decision variables $\epsilon_2$, $Y$, $P$. 
A line search with respect to $\epsilon_1$ on top of solving this LMI is typically preferable than solving directly the bilinear matrix inequality in~\eqref{qlmi data}. Note also that model-based approaches for controlling  bilinear systems encounter such a situation, and 
fix one of the parameters directly \cite{Khlebnikov2018arc} or in an iterative way \cite{tarbouriech2009control}.

A conclusion of Proposition~\ref{prop:data based quadr stab} is that the basin of attraction of the origin contains the set $\mathcal{E}_{Q}=\mathcal{E}_{P^{-1}}$. It is quite natural to maximize the volume of this ellipsoid, which is proportional to $\det (P)$, as is done in the model-based setting of \cite{Khlebnikov2018arc}. (Other size criteria can be optimized, see the discussion in \cite[\S 2.2.5.1]{tarbouriech2011stability}.) This leads to the next immediate corollary.
\begin{corollary}
\label{cor:maximize}
Let Assumption~\ref{ass:matrixNormD} hold.
If there exist a solution to the next optimization problem in the decision variables $\epsilon_1 \in \real$, $\epsilon_2 \in \real$, $Y \in \real^{n \times T}$ and $P=P^\top \in \real^{n \times n}$
\begin{equation*}
\begin{aligned}
& \text{minimize } && -\log \det (P)\\
& \text{subject to} &&  \eqref{qlmi data}, \eqref{qlmi data:I=X0 GK},\\
\end{aligned}
\end{equation*}
then the conclusion of Proposition~\ref{prop:data based quadr stab} holds.
\end{corollary}

Finally, since we are considering a quadratic Lyapunov function and as is done in the model-based solutions \cite{Khlebnikov2018arc,tarbouriech2009control}, the very same arguments leading to Proposition~\ref{prop:data based quadr stab} yield exponential (instead of asymptotic) stability by strengthening a little the matrix inequality in~\eqref{qlmi data}. This is stated in the next corollary, whose proof is thus omitted.

\begin{corollary}
For $\mu \in (0,1)$, suppose that the assumptions of Proposition~\ref{prop:data based quadr stab}
can be satisfied after replacing the element $(1,1)$ of the matrix in~\eqref{qlmi data} (i.e., $-P$) with $- \mu P$. Then,
\begin{enumerate}[wide,labelindent=0pt,label=(\roman*)]
\item for the dynamics in~\eqref{bilinear data D} corresponding to $D$, the Lyapunov function $V(x)= x^\top Q x= x^\top (X_{0,T} Y^\top)^{-1} x$ satisfies 
\begin{equation*}
V(g_D(x) x) < \mu V(x) \quad \text{ for all } x \in \mathcal{E}_{Q}\backslash\{0\};
\end{equation*}
\item the origin is exponentially stable for \eqref{bil.sys} with controller $u=K x = U_{0,T} G_K x = U_{0,T} Y^\top (X_{0,T} Y^\top)^{-1} x$ and its basin of attraction contains the set $\mathcal{E}_{Q}$.
\end{enumerate}
\end{corollary}

\begin{figure}
\centerline{\hspace*{.0cm}\includegraphics[scale=.65]{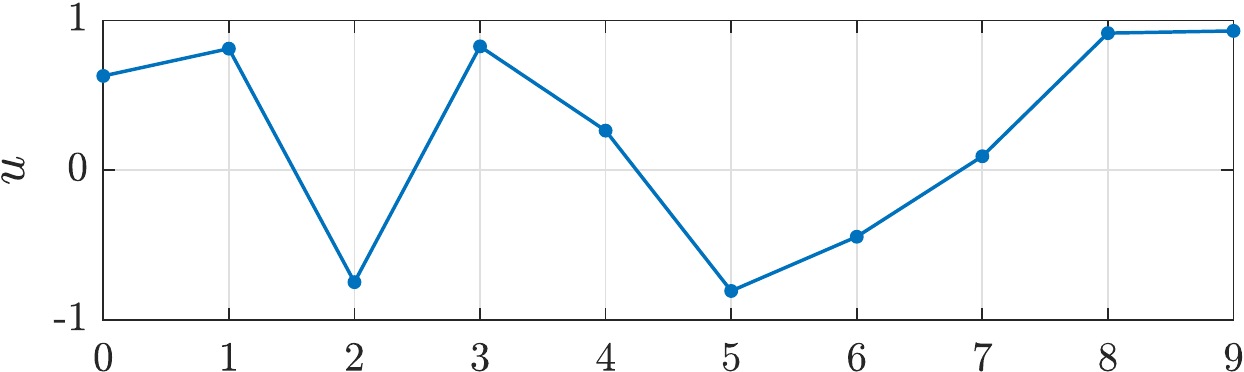}}
\centerline{\includegraphics[scale=.65]{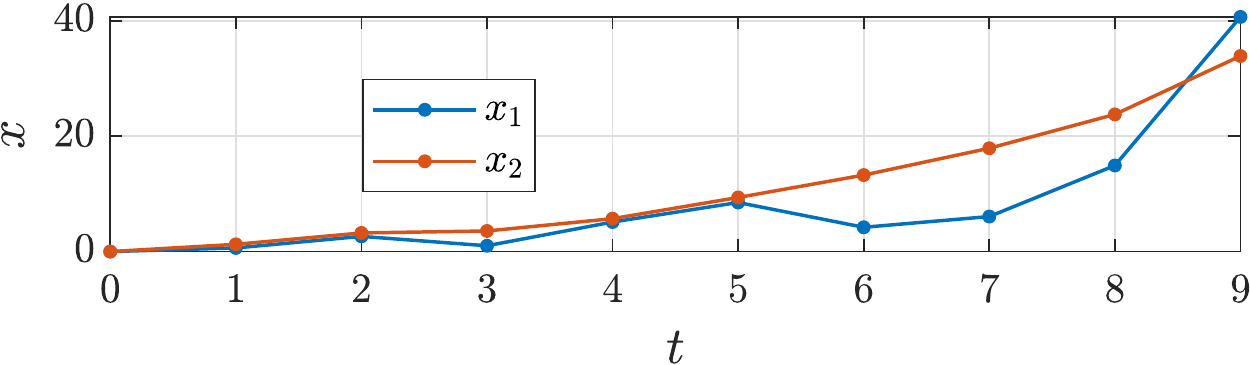}}
\caption{Input and state sequences giving the quantities in~\eqref{data}.}
\label{fig:data}
\end{figure}

\section{Numerical example}
\label{sec:numerical example}

We consider for~\eqref{bil.sys} the matrices
\begin{equation}
\label{bil.sys.matrices}
A=\begin{bmatrix}
0.8 & 0.5\\
0.4 & 1.2
\end{bmatrix},
B=\begin{bmatrix}
1\\
2
\end{bmatrix},
D=\begin{bmatrix}
0.45 & 0.45\\
0.3 & -0.3
\end{bmatrix},
\end{equation}
which are taken from~\cite[\S 5]{bitsoris2008ifac}.
Our design does \emph{not} rely on their knowledge, but simply on the data generated according to them and a bound $\delta$ of $\| D \|$. In particular, we consider $\delta=0.7637$, which overapproximates by 20\% the actual $\| D \|= 0.6364$. Moreover, we will use the matrices in~\eqref{bil.sys.matrices} to compare our data-based design with a model-based design in Section~\ref{sec:numerical example:MB}.
We note that the comparison is made with a model-based design that has perfect knowledge of the parameters of the system. Getting to perfectly know the parameters would correspond to the ideal case even for a preliminary system identification step. We show in this section that our designed controller performs comparably to such a model-based design, in spite of being tuned only on an offline experiment.

We consider $T=10$. In Figure~\ref{fig:data}, we show the input and state sequences giving \eqref{data} and generated according to the matrices in~\eqref{bil.sys.matrices}.
We note that $A$ being unstable is challenging because a suitable control action has to be designed to modify by feedback the system evolution in a neighborhood of the origin (without the ``help'' of a stable linear part) and the generated data quickly grow large as shown in Figure~\ref{fig:data}, thereby impacting the numerical accuracy of the procedure.

\subsection{Data-based solution}
\label{sec:numerical example:DB}

By using Corollary~\ref{cor:maximize}, the data-based solution implemented in this section is as follows. In particular, we opt for fixing the scalar variable $\epsilon_1$, solve an LMI, and perform a line search on $\epsilon_1$.
\begin{enumerate}[leftmargin=12pt]
\item We fix $\epsilon_1>0$.
\item We solve the next optimization problem in the decision variables $\epsilon_2 \in \real$, $Y \in \real^{n \times T}$ and $P=P^\top \in \real^{n \times n}$
\begin{equation*}
\begin{aligned}
& \text{minimize } -\log \det (P)\\
& \text{subject to }  P=X_{0,T} Y^\top,\\ 
& \begin{bmatrix}
- P						& 0						& Y U_{0,T}^\top 			& Y X_{1,T}^\top 			& -\delta Y V_{0,T}^\top \\
\star					& - \epsilon_1 P 		& 0							& 0							& \delta \epsilon_1 P\\
\star					& \star					& -\epsilon_1 I				& 0							& 0\\
\star					& \star					& \star						& - P + \epsilon_2 I		& 0\\
\star					& \star					& \star						& \star						& - \epsilon_2 I
\end{bmatrix} \prec 0,
\end{aligned}
\end{equation*}
which is an LMI. By denoting the solution $P=:P_\textup{DB}$, we then obtain $G_K= Y^\top P^{-1}_\textup{DB}$ and the controller gain as $K_\textup{DB}:= U_{0,T} G_K$.
\item We iterate on the selection of $\epsilon_1$ in case of, e.g., infeasibility.
\end{enumerate}

We implement this scheme (and the model-based one in Section~\ref{sec:numerical example:MB}) through YALMIP \cite{lofberg2004yalmip}. 
For a value of $\epsilon_1= 0.8$, we obtain 
\begin{equation*}
P_\textup{DB}\!=\!
\begin{bmatrix}
    3.2827 &   -0.9642\\
   -0.9642 &   2.4388\\
\end{bmatrix}\!,\,
K_\textup{DB}\!=\!
\begin{bmatrix}
-0.3175  & -0.5649
\end{bmatrix}.
\end{equation*}
The evolution of $x$ when $u= K_\textup{DB} x$ is used in~\eqref{bil.sys} is given in Figure~\ref{fig:DBandMBsolution} in the top plot as a phase portrait (\emph{solid} colored lines) and in the middle plot as a time evolution.

\subsection{Comparison with model-based solution}
\label{sec:numerical example:MB}

For~\eqref{bil.sys} with matrices in~\eqref{bil.sys.matrices}, we use the model-based solution in~\cite{Khlebnikov2018arc} for comparison. This model-based solution is also not an LMI, unless the scalar parameter $\epsilon_1$ is fixed (as in the data-based solution) and a line search is performed.

\begin{enumerate}[leftmargin=12pt]
\item We fix $\epsilon_1>0$.
\item We solve  the optimization problem in the decision variables $y \in \real^{n}$ and $P=P^\top \in \real^{n \times n}$
\begin{equation*}
\begin{aligned}
& \text{minimize } -\log \det (P)\\
& \text{subject to } P\succ 0\\ 
& \begin{bmatrix}
- P					& 0					& y 				& P A^\top + y B^\top\\
0					& - \epsilon_1 P 	& 0					& P D^\top\\
y^\top 				& 0					& -\epsilon_1 I		& 0\\
A P +B y^\top		& D P				& 0					& - P\\
\end{bmatrix} \prec 0,
\end{aligned}
\end{equation*}
which is an LMI.
By denoting the solution $P=:P_\textup{MB}$, we then obtain the controller gain as $K_\textup{MB}:= y^\top P^{-1}_\textup{MB}$.
\item We iterate on the selection of $\epsilon_1$ in case of, e.g., infeasibility.
\end{enumerate}

For the same value of $\epsilon_1= 0.8$ as in Section~\ref{sec:numerical example:DB}, we obtain 
\begin{equation*}
P_\textup{MB}\!=\!
\begin{bmatrix}
    8.5623 &   -4.7253\\
   -4.7253 &   6.3616\\
\end{bmatrix}\!,\,
K_\textup{MB}\!=\!
\begin{bmatrix}
-0.3572  & -0.5738
\end{bmatrix}.
\end{equation*}
The evolution of $x$ when $u= K_\textup{MB} x$ is used in~\eqref{bil.sys} is given in Figure~\ref{fig:DBandMBsolution} in the top plot as a phase portrait (\emph{dotted} colored lines) and in the bottom plot as a time evolution. 

\begin{figure}[t]
\centerline{\includegraphics[scale=.65]{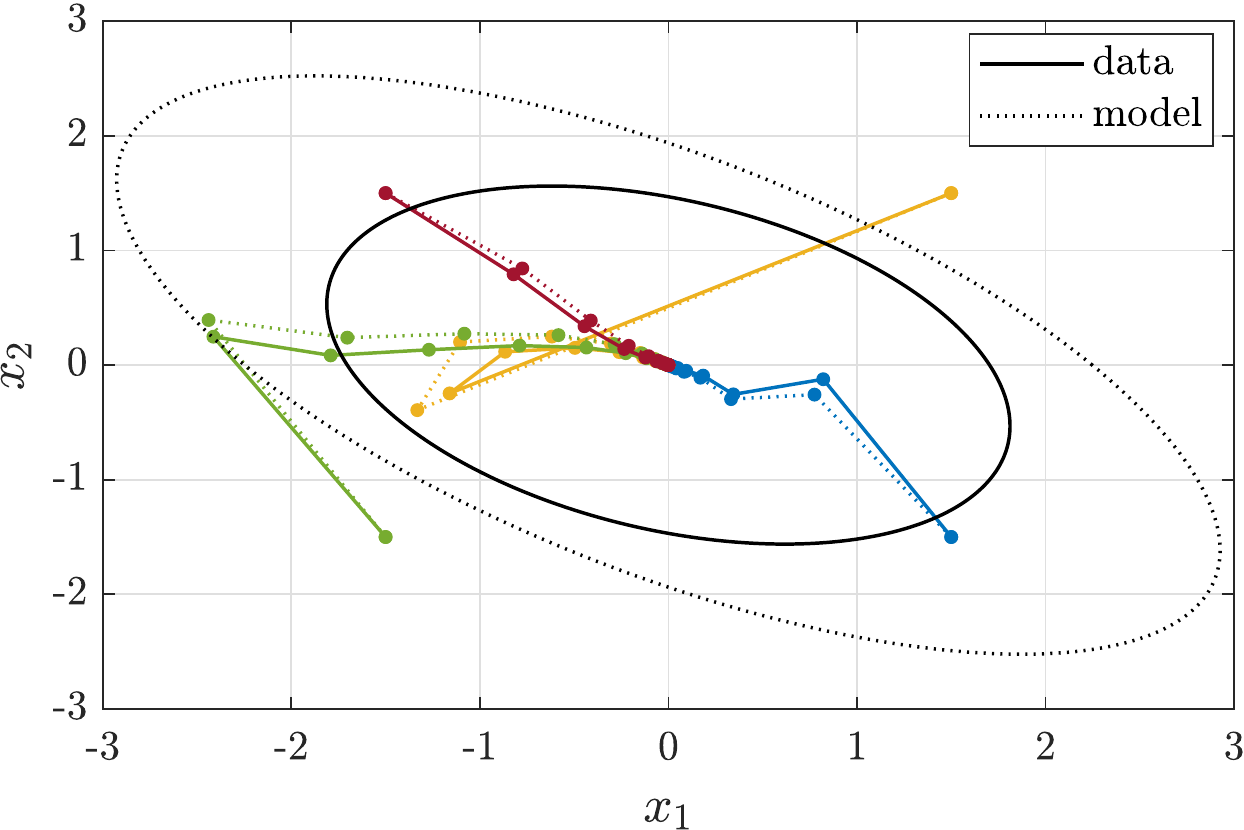}}
\centerline{\includegraphics[scale=.65]{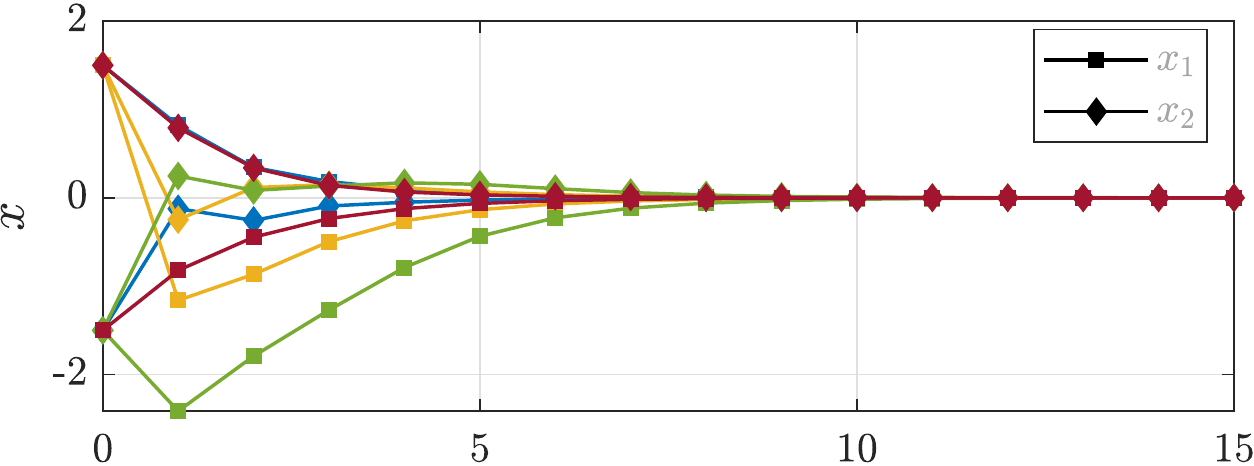}}
\centerline{\includegraphics[scale=.65]{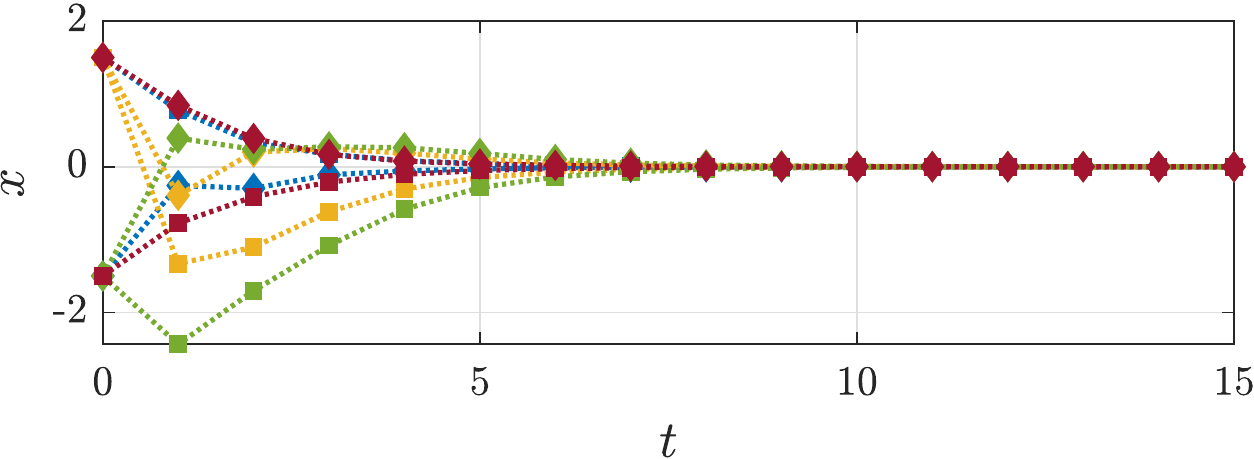}}
\caption{Evolution of the data-based and model-based solutions of Sections~\ref{sec:numerical example:DB} and \ref{sec:numerical example:MB}, corresponding to the selected value of $\epsilon_1$. The same color corresponds to solutions with the same initial condition. Solid and dotted lines correspond respectively to the data-based and model-based solution.
(Top) Phase portrait. The area within the ellipsoids is guaranteed to be in the basin of attraction of the origin, by the existence of the Lyapunov functions corresponding to the matrices $P_\textup{DB}$ and $P_\textup{MB}$.
(Middle) Time evolutions of the state $x$ for the data-based solution, where traces with squares and diamonds identify respectively the component $x_1$ and $x_2$.
(Bottom) Time evolutions of the state $x$ for the model-based solution.
}
\label{fig:DBandMBsolution}
\end{figure}

\subsection{Discussion}

Finally, we compare the performance of the data-based solution against the model-based solution by performing a thorough line search on the parameter $\epsilon_1$, which we fixed before in order to be able to solve an LMI. The result is in Figure~\ref{fig:line search}. 
Only values of $\epsilon_1$ where an optimal solution was returned by YALMIP, are displayed (in particular, this did not happen for the model-based solution with values of $\epsilon_1$ between $0.2$ and $0.4$).

The top plot represents the determinants of the matrices $P_\textup{DB}$ and $P_\textup{MB}$, which was considered since it is proportional to the volume of the ellipsoids that are guaranteed to be in the basin of attraction of the closed-loop system. 

In the middle plot, the logarithms of these determinants are also provided since they are the actual objective functions in the optimization problems of Sections~\ref{sec:numerical example:DB}-\ref{sec:numerical example:MB}. 

As expected, the model-based solution provides ellipsoids with larger sizes (e.g., $\det(P_\textup{MB})=60.03$ for $\epsilon_1=0.4$). 
For the given example, it appears from Figure~\ref{fig:line search} that the data-based solution performs better for small $\epsilon_1$, whereas it performs worse than the model-based solution for large $\epsilon_1$.
We note that $\log \det$ is actually more representative of the actual difference between the two solutions. Indeed, for values of $\epsilon_1$ around $1$, the two solutions are not so distant, as is confirmed by the illustration of Figure~\ref{fig:DBandMBsolution} where the corresponding ellipsoids are also depicted in the top plot (solid and dotted black curves).

Finally, since the model-based and data-based solution share a similar structure, we expect that they lead to similar feedback gains. This is indeed confirmed by their relative difference in norm in the bottom plot in Figure~\ref{fig:line search}.

In summary, our designed controller presents in these simulations a similar performance to the model-based design, where the former relies on an offline experiment and the latter on the perfect knowledge of system parameters.

\begin{figure}
\centerline{\includegraphics[scale=.65]{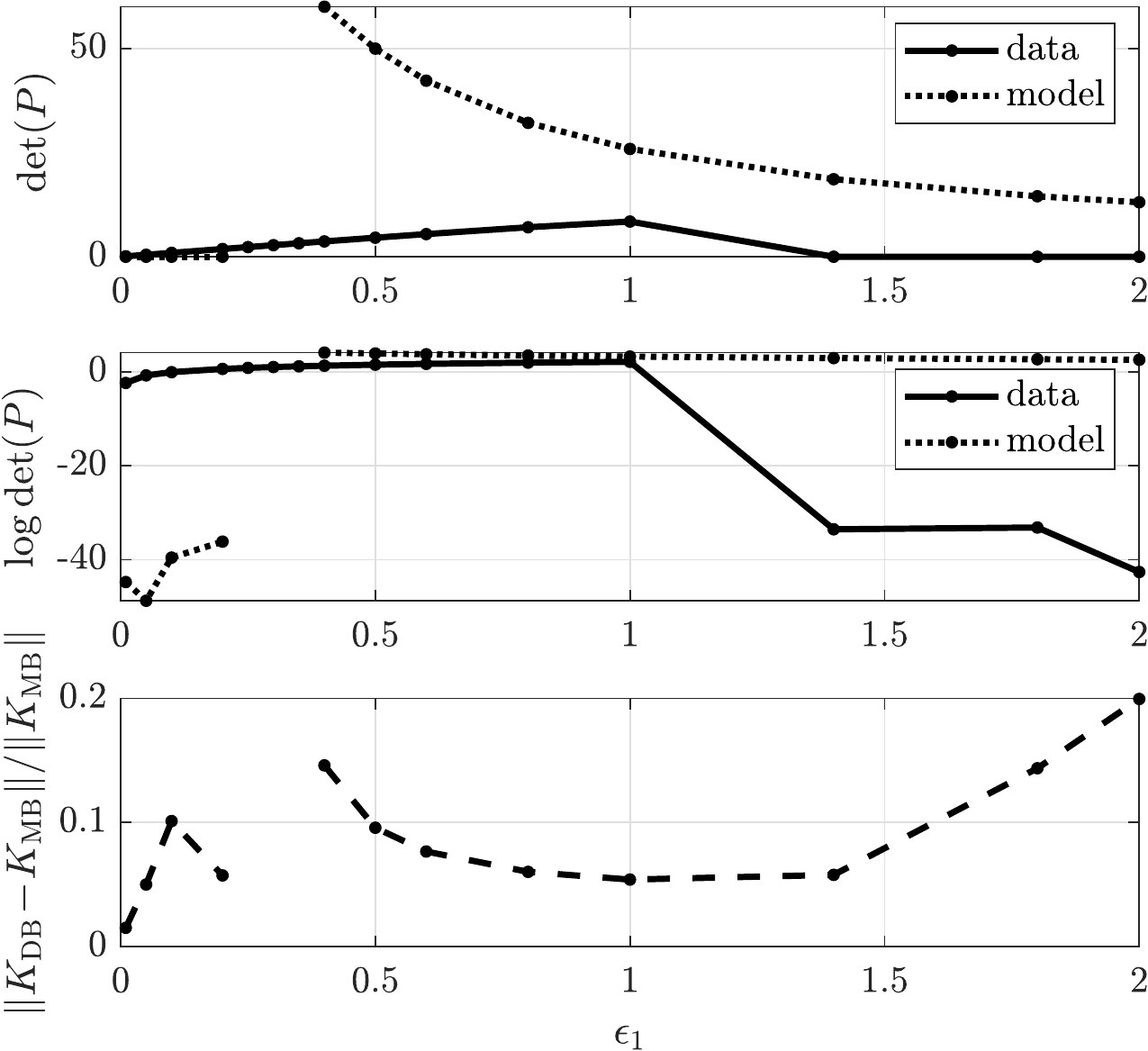}}
\caption{
Characterization of the main parameters of the data-based and model-based solution as a function of the parameter $\epsilon_1$. 
(Top) The determinants of matrices $P_\textup{DB}$ and $P_\textup{MB}$.
(Middle) Their logarithms, which are optimized in the numerical schemes of Section~\ref{sec:numerical example:DB} and \ref{sec:numerical example:MB}.
(Bottom) The relative difference in norm between the feedback gains.
}
\label{fig:line search}
\end{figure}

\section{Conclusions}

We proposed a direct data-driven design for bilinear systems, which comes with a guaranteed subset of the basin of attraction. 

The main goal of future work is applying this scheme as a building block for data-driven control of input-affine nonlinear systems (by approximating the latter through Carleman linearization). Closely related topics of future work are a study of the modifications needed to cope with noisy data,
and of the tradeoffs with schemes based on sum-of-squares programming for bilinear systems.

\bibliographystyle{plain}
\bibliography{biblio-data}

\end{document}